%% file: ms.tex
\documentclass[sigconf]{acmart}
\usepackage{booktabs} 
\usepackage{epsfig}\usepackage{epstopdf}
\usepackage{graphicx,subfigure}
\usepackage{natbib}

\RequirePackage{fancyhdr}

\begin{document}

\copyrightyear{2017}
\acmYear{2017}
\setcopyright{acmcopyright}
\acmConference{NANOCOM '17}{September 27--29, 2017}{Washington D.C.,
DC, USA}\acmPrice{15.00}\acmDOI{10.1145/3109453.3109466}
\acmISBN{978-1-4503-4931-4/17/09}

\title[]{ Received Signal Strength for Randomly Distributed Molecular Nanonodes}

\author{Rafay Iqbal Ansari \\ Chrysostomos Chrysostomou}

\affiliation{%
  \institution{Department of Computer Science and Engineering, Frederick University}
  \city{Nicosia} 
  \state{Cyprus} 
  \postcode{1036}
}
\email{rafay.ansari@stud.frederick.ac.cy, ch.chrysostomou@frederick.ac.cy}

\author{Taqwa Saeed \\ Marios Lestas}
\affiliation{%
  \institution{Department of Electrical Engineering, Frederick University}
  \city{Nicosia} 
  \state{Cyprus} 
  \postcode{1036}
}
\email{taquasalaheldin@gmail.com,
 eng.lm@frederick.ac.cy}

\author{Andreas Pitsillides}
\affiliation{%
  \institution{Department of Computer Science, University of Cyprus}
  \city{Nicosia} 
  \country{Cyprus}}
\email{andreas.pitsillides@ucy.ac.cy}

%
\begin{abstract}
We consider nanonodes randomly distributed in a circular area and characterize the received signal strength when a pair of these nodes employ molecular communication. Two communication methods are investigated, namely free diffusion and diffusion with drift. Since the nodes are randomly distributed, the distance between them can be represented as a random variable, which results in a stochastic process representation of the received signal strength. We derive the probability density function of this process for both molecular communication methods. Specifically for the case of free diffusion we also derive the cumulative distribution function, which can be used to derive transmission success probabilities. The presented work constitutes a first step towards the characterization of the signal to noise ratio in the considered setting for a number of molecular communication methods.
\end{abstract}
%
%
%
%
%
\maketitle
\renewcommand{\shortauthors}{Rafay et al.}
\input{samplebody-conf}
\bibliography{reff} 
\bibliographystyle{ACM-Reference-Format}
\end{document}

%% file: samplebody-conf.tex
\section{Introduction}

Nanonetwork communication can be molecular based, electromagnetic based or hybrid between the two \cite{misra2014green}. Molecular Communication (MC) has a number of advantages, as for example, biocompatibility, scalability and energy efficiency \cite{paradigmfernando}. MC is based on the encoding and decoding of chemical signals, where molecules are used as information carriers \cite{surveyonmolec15}. MC is nature inspired, as it is the preferred communication method for many kinds of living cells in both macro and micro scales. There are a number of ways with which molecules can be transferred from the transmitter to the receiver. The easiest way is free diffusion, where the transmitter simply diffuses the information molecules into the environment "hoping" that they (or at least part of them) will arrive at the receiver to convey the intended message. When information molecules are released, they move based on Brownian Motion using thermal energy available in the surrounding environment \cite{brownianeck}, which is a relatively slow process. A faster method is diffusion with drift, which involves flows that direct and speed up the movement of the molecules towards the receiver. This communication method is observed in the human body, as for example when hormones (information molecules) are diffused into the blood stream, which absorbs them where needed. In \cite{cobo2010bacteria}, the authors propose the design of a bacteria-based nanonetwork, where messages between nanodevices are carried by \textit{Escherichia coli}. In this type of networks, the message is encoded in double-stranded DNA molecules inside the transmitter. The DNA molecule comprises of a transfer region, a routing region and a message region. The transfer and routing regions contain information related to self replication and transmission, programmable death, self-replication inhabitation and behavioral differences between empty and laden bacteria carriers. The message region includes destination address and the message body. The transmitter releases a certain type of attractant molecules to attract \textit{empty} bacteria carriers and passes a copy of the DNA to the bacteria through conjugation. Then, the receiver nanodevice diffuses molecules, which attract laden bacteria such that the receiver location is where the concentration of attractants reaches its peak value. The bacteria swims towards the area with higher concentration of attractants using flagella, until it reaches the receiver nanonode, to  which a copy of the DNA is then passed through conjugation. Finally, the receiver destroys the bacterium after decapsulating the DNA to prevent the former of redelivering the message. It is also worth mentioning that the bacterium which does not deliver its DNA within a predefined time, suicides to prevent delivering overdue messages.
Another molecular-based method to transfer data is using Molecular Motors (MM). In nature, protein molecular motors like Kinesins transfer cargoes (e.g. lipid vesicles) inside living cells. However, researchers have been able to study Kinesin motors \textit{in vitro} environments \cite{svoboda1994force,hiyama2010biomolecular}. Therefore, molecular motors are considered as a molecular propagation method for synthetic molecular nanonetworks \cite{walking,surveyonmolec15}. MM move over microtubules, which are tubular bodies of diameter $24 nm$ that form a track between the source and destination. The motor's structure consists of two heads, two necks, a stalk and two tails. The tails hold the cargo while the heads are attached to the microtubules. The way the movement is achieved is the following: one of the heads uses hydrolysis of adenosine triphosphate (ATP)  to release energy to detach from the microtubule track, moves forwards in a certain direction, and attach to the microtubule in the new location. The second head performs the same process in the same direction. The repetition of this process causes the "walking" effect of the MM over microtubules tracks \cite{surveyonmolec15}. Furthermore, there is another synthetic configuration with which MM can be used to transport cargoes. In \cite{hiyama2010biomolecular}, it has been indicated that microtubules (MTs) can glide over a surface coated with stationary kinesin to transport cargoes attached to MTs. The authors propose an autonomous loading, transport and unloading system, where cargo-liposomes are labeled with a single stranded DNA (ssDNS) that is complementary to the ssDNA labeling the MTs. So, 15-base ssDNA MT loads 23-base DNA labeled liposomes, diffused at the loading site, through the process of DNA hybridization. The MT then continues gliding over the kinesin-coated surface until it reaches the unloading site, where 23-base ssDNAs that are complimentary to the DNA of the liposome are microarrayed. Then the MT unloads the cargo-liposome through DNA hybridization and continues to glide. Intercellular Calcium Waves (ICW) is a molecular communication used among biological cells. When a cell is triggered, it responds by utilizing 1,4,5-triphosphate ($IP_{3}$) to increase the concentration of cytosolic $Ca^{2+}$. Adjacent cells share a gap of an adjustable permeability to pass selected molecules between cells. Those gaps are the internal way of communicating $IP_{3}$ molecules, which generate the increase in $Ca^{2+}$. However, the ICW can also be communicated externally using ATP molecules, which are released into the surrounding environment to stimulate cells in the vicinity \cite{kuran2012calcium}. The exchanged information is encoded in the amplitude and concentration of the ICW \cite{schuster2002modelling}. A basic design of a molecular communication system based on ICW is proposed in \cite{nakano2005molecular}. Another way of communication among living cells is through Neurochemical propagation. This type of communication is carried out by neurons, either to another neuron or to a junction cell such as in neuromuscular junctions \cite{surveyonmolec15}. The communication is established by diffusing neurotransmitter vesicles by synapse into the synaptic cleft. The vesicles then bond with certain receptors within the membrane of the receiver cell.\\
Significant work has been conducted to model and characterize the physical layer properties of the aforementioned molecular communication methods. Such characterizations are directly related to the transceiver design, the utilized modulation techniques and detection methods, which lead to characterizations of the achieved transmission success probabilities and channel capacities \cite{pamnppm, guo2016molecular, srinivas2012molecular,mahfuz2011detection}. Many of these attempts are diffusion based. The authors of \cite{kuran2010energy}, for example, characterize the capacity of the molecular communication channel based on a binomial distribution model of the received molecules. They also find the probability of a successful reception based on the distribution of the previous and current sent bit. The outage probabilities are well known to be largely dependent on representations of the received signal strength at the receiver and most importantly on the signal to noise (SNR) ratios. Attempts to obtain such representations were primarily based on independent models of the noise behavior \cite{farsad2014channel,pamnppm}. A notable work, which presents analysis of the received signal strength in case of multiple transmissions is reported in \cite{pierobon2014statistical}. Therein randomness in both the locations of the nodes and the transmitted signal are considered, however, this allows only for a closed form expression for the Probability Density Function (PDF) of the received Power Spectral Density (PSD) to be derived using the log-characteristic function. It is reported therein that the PDF of the received signal is not possible to derive.
In this work, we relax the randomness in the transmitted signal and we consider the communication of a pair of nodes randomly distributed in a circular area of radius $r$. We consider two communication methods, namely free diffusion and diffusion with drift. For the former case we derive both the PDF and the cumulative distribution function (CDF) of the received signal strength. The CDF can be used to derive successful transmission probabilities for different transmitted signal strengths, thus being useful in deciding the transmitted signal strength that would lead to specific quality of service guarantees. For the case of diffusion with drift, we derive the probability density of the received signal. This work constitutes a first step in deriving characterizations of the signal to noise for various molecular communication methods.         
The paper is organized as follows, in Section \ref{stat} we describe the available physical channel models for free diffusion and diffusion with drift and then present the PDF and CDF for each model, Section \ref{sim} includes MATLAB simulations that validate the derived expressions, and finally in Section \ref{conc} we offer our conclusion and future directions. 

\section{Modeling and Analysis} \label{stat}
We consider a pair of nanonodes randomly distributed in a circular area of radius $r$. The nodes employ molecular communication, one acting as the transmitter and the other one acting as the receiver as shown in Fig. \ref{system}. Two methods of molecular communication are considered, namely free diffusion and diffusion with drift. The transmitter is assumed to transmit an impulse signal $g(t)$ of intensity $M$ such that: 
\begin{equation}
\int_{-\infty}^{\infty} g(t) dt = M .
\end{equation}
Since the nodes are randomly distributed in the considered area, the distance $X$ between them is a random variable (RV). The received signal strength $Y$ at the receiver is also a RV. The objective is to derive the properties of this random variable, namely, the probability density function and the cumulative distribution function. Below, we present the obtained results for the two considered communication models.     
\begin{figure}[h]
	\centering
		\includegraphics[width=4 cm]{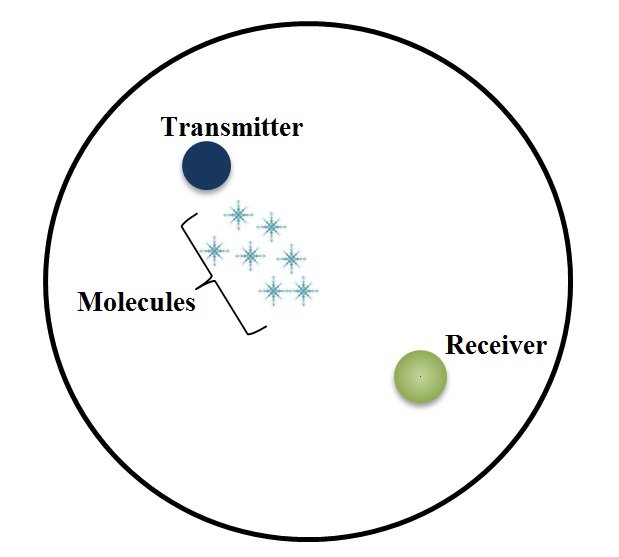}
	\caption{The considered system}
	\label{system}
\end{figure}
\subsection{Free Diffusion}
It is well established in the molecular communications literature that Greens' function can be used to represent the impulse response of the received signal strength at a time $t$ at distance $x$ from the point where the impulse signal is transmitted. The impulse response is given by the following:      
\begin{equation}\label{diffu}
y=M\frac{e^{-\frac{x^2}{4Dt}}}{4 \pi Dt},
\end{equation}
where $D$ is the diffusion coefficient and $M$ is the intensity of the impulse signal (number of transmitted molecules in one short burst). Since the transmitter and the receiver are randomly distributed, the distance $x$ between them is a RV which we denote by $X$. It has been established in \cite{moltchanov2012distance} that the PDF of the distance between two randomly distributed nodes in a circular region of radius $r$ is given by:
\begin{align} \label{eucd}
f_X(x)=\frac{2x}{r^{2}}\left(\frac{2}{\pi}\arccos\left(\frac{x}{2r}\right)-\frac{x}{\pi r}\sqrt{1-\frac{x^2}{4 r^2}}\right).
\end{align}
In the theorem below, we present an expression for the PDF of the received signal strength.
\begin{theorem} \label{thro}
The PDF of the received signal strength $Y$ when two nodes randomly distributed in a circular area of radius $r$ employ molecular communication with impulse response (\ref{diffu}) is given by:
\begin{align}\label{pdf1}
\begin{split}
&h_Y(y)=\frac{4Dt}{\pi r^2 y} \left[2 arccos\left(\frac{(-4Dt \log{(\frac{4 \pi D t y}{M})})^{1/2}}{2r}\right)\right.-\\ 
 &\left.\frac{(-4Dt \log{(\frac{4 \pi D t y}{M})})^{1/2}}{r} \sqrt{1-\frac{(-4Dt \log{(\frac{4 \pi D t y}{M})})}{4r^2}}\right].
\end{split}
\end{align}
\begin{proof}
The PDF of the Euclidean distance between the transmitter and the receiver is given by (\ref{eucd}),
where $0\leq x \leq2r$. In the concentration formula (\ref{diffu}), the square of the Euclidean distance is used. Therefore, we find the distribution of the RV $Z$, where $Z= X^{2}$. The distribution of the Euclidean distance to an arbitrary power $\beta$ is given as: 
\begin{align}\label{veta}
\begin{split}
f_Z(z)&=\frac{1}{\beta}z^{(\frac{1}{\beta}-1)}{f_X(z^{(\frac{1}{\beta})})}.
\end{split}
\end{align}
In our case $\beta = 2$. So, the PDF of $Z$ is expressed as:
\begin{align}\label{ZPDF}
\begin{split}
  f_Z(z) =\frac{1}{\pi r^2} \left(2 \arccos\left(\frac{z^{\frac{1}{2}}}{2r}\right)- \frac{z^{\frac{1}{2}}}{r}\sqrt{1- \frac{z}{4r^2}} \right) ,
\end{split}
\end{align}
where $0\leq z \leq (2r)^{2}$.
Hence, the transformation function becomes:
\begin{align}
\begin{split}
y(z)=\frac{M e^{-\frac{z}{4Dt}}}{4 \pi D t}
\end{split}
\end{align}
Let $Y=y(Z)$ be the derived RV and $h(y)$ be the PDF of the derived RV. We now apply the transformation of the density function to determine $h(y)$, which is defined as:
\begin{align}\label{befin}
\begin{split}
h_{Y}(y)=f_{Z}(z(y)) \left|\frac{dz}{dy}\right|,
\end{split}
\end{align}
where $\frac{dz}{dy}=\frac{-4 Dt}{y}$. 
In order to find the PDF of $h(y)$ we find the inverse of $y(z)$, which results in $z(y)=-4Dt \log{(4 \pi D t y)}$. Accordingly, we incorporate $z(y)$ in (\ref{befin}) to find the expression in (\ref{pdf1}).
\end{proof}
\end{theorem} 
We now use the PDF given in (\ref{pdf1}) to find the CDF of the concentration, which is represented by the RV $Y$.
\begin{theorem}
The CDF of the RV $Y$ is given by: 
\begin{multline}\label{CDF11}
H_Y(y)=\frac{4 D t}{r^2} (\log{(y)} - \log{(\frac{M}{4 \pi Dt})}) \\ - \frac{8Dt}{\pi r^{2}} \sum_{n=0}^{\infty} \sigma_n - \left( \frac{2^{2n+2} (-Dt \log{(\frac{4 \pi D t y}{M})})^{\frac{3}{2}+n}}{Dt(3+2n)}\right)
\\ -\frac{4 D t}{\pi r^3} \sum_{0}^{\infty} \gamma_{n} -\left(\frac{2^{2+2n}(-Dt \log{(\frac{4 \pi D t y}{M})})^{\frac{3}{2}+n}} {D(3+2n)t}\right)
\end{multline}
\begin{proof}
To find the CDF we integrate the PDF
\begin{equation}
\begin{split}
H_Y(y)&= \int_{-\infty}^{y} h_Y(y) dy = \mathcal{I}_1-\mathcal{I}_2
\end{split}
\end{equation}
For the ease of analysis we have split the integral into $\mathcal{I}_1$ and $\mathcal{I}_2$.
\begin{equation}
\mathcal{I}_1=\frac{8 D t}{\pi r^2}\int_{-\infty}^{y} \frac{1}{y} arccos\left(\frac{(-4 D t \log{(\frac{4 \pi D t y}{M})})^{1/2}}{2r}\right)dy
\end{equation}
To solve $\mathcal{I}_1$ we use the power series expansion $\arccos(c)=\frac{\pi}{2}-\sum_{n=0}^{\infty}\frac{\binom{2n}{n} c^{2n+1}}{4^n(2n+1)}$, for $\quad |c|<1$, yields:
\begin{multline}
\mathcal{I}_1 =\frac{4 D t}{r^2}\int_{\frac{M}{4 \pi D t}}^{y}\frac{1}{y}dy\\ 
  - \frac{8Dt}{\pi r^{2}} \sigma_{n}\int_{\frac{M}{4 \pi D t}}^{y} \frac{1}{y} (-4 D t)^{\frac{2n+1}{2}} (\log{(\frac{4 \pi D t y}{M})})^{\frac{2n+1}{2}}\\
=\frac{4 D t}{r^2} (\log{(y)} - \log{(\frac{M}{4 \pi Dt})}) \quad \quad  \quad \quad \quad \quad \quad \\ - \frac{8Dt}{\pi r^{2}} \sum_{n=0}^{\infty} \sigma_n - \left( \frac{2^{2n+2} (-Dt \log{(\frac{4 \pi D t y}{M})})^{\frac{3}{2}+n}}{Dt(3+2n)}\right),
\end{multline}
where $\sigma_n\overset{\Delta}{=}\frac{\binom{2n}{n} c^{2n+1}}{4^n(2n+1)}$. It is worth noting that the lower limit of the integration is found by evaluating (\ref{diffu}) when $x=0$. \\
$\mathcal{I}_2$ is defined as follows:
\begin{multline}
\mathcal{I}_2=\frac{4 D t}{\pi r^3} \sum_{0}^{\infty} \gamma_{n}
 \int_{\frac{M}{4 \pi D t}}^{y}  \left(-4Dt \log{(\frac{4 \pi D t y}{M})} \right)^{\frac{1}{2} + n} \frac{1}{y}\\
=\frac{4 D t}{\pi r^3} \sum_{0}^{\infty} \gamma_{n} -\left(\frac{2^{2+2n}(-Dt \log{(\frac{4 \pi D t y}{M})})^{\frac{3}{2}+n}} {D(3+2n)t}\right)
\end{multline}
We have used the binomial series expansion 
\begin{equation}
\sqrt{1-\frac{z}{4 r^2}}=\sum_{0}^{\infty} \gamma_n z^{n},
\end{equation}
where the coefficient $\gamma_{n} = \binom{1/2}{n}  \frac{(-1)^n}{4^n r^{2n}}.$ The final result is given in (\ref{CDF11}).
\end{proof}
\end{theorem}
The CDF derived above can be used to find the success probability when the success threshold, $\tau$, is defined. 
\subsection{Diffusion with Drift}
Similar analysis is done to find the PDF of the concentration when diffusion with drift is employed. The concentration in this case is modeled by Wiener's equation \cite{srinivas2012molecular}:
\begin{equation} \label{diffdrift}
y=\frac{M}{\sqrt{4 \pi D t}}e^{-\frac{(x-vt)^2}{4Dt}},
\end{equation}
where $M$ is the number of molecules released during the short burst, $X$ is a RV representing the Euclidean distance between the transmitter and the receiver, $v$ is the average flow speed and $t$ is the time.  
\begin{theorem}
If the transmitter and receiver exist in a circular region of radius $r$ and are separated by an arbitrary Euclidean distance bounded by the diameter of the circle, the PDF of the diffusion with drift function $y$ is given as:
\begin{align}\label{drifft}
\begin{split}
h_Y(y)&=\frac{(1+(-4 D t log(\sqrt{4 \pi D t} \frac{y}{M})))^{\frac{-1}{2}}vt}{r^2}\\
&\left(\frac{2}{\pi} arccos\left(\frac{(-4 D t log(\sqrt{4 \pi D t} \frac{y}{M}))^{\frac{1}{2}}+vt}{2r}\right)\right.-\\
&\left(\frac{(-4 D t log(\sqrt{4 \pi D t} \frac{y}{M}))^{\frac{1}{2}}+vt}{2r}\right)\\
&\sqrt{1-\frac{((-4 D t log(\sqrt{4 \pi D t} \frac{y}{M}))^{\frac{1}{2}}+vt)^{2}}{4r^2}}.\\
\end{split}
\end{align}
\end{theorem}
\begin{proof}
The PDF of the Euclidean distance between the nanodevices uniformly distributed in a circle of radius $r$ is given in (\ref{eucd}). First, a linear transformation $u(x)=x-vt$ is applied on the RV $X$, resulting in
\begin{equation}
\begin{split}
h_{U}(u)&=f_{(u(x))}\left|\frac{du}{dy}\right|\\
&=\frac{2(u+vt)}{r^{2}}\left(\frac{2}{\pi}\arccos\left(\frac{(u+vt)}{2r}\right)-\right.\\ 
&\left.\frac{(u+vt)}{\pi r}\sqrt{1-\frac{(u+vt)^2}{4 r^2}}\right).
\end{split}
\end{equation}
In the concentration formula (\ref{diffdrift}) the square of $u$ is used. Therefore, we find the distribution of the RV $Z$ using (\ref{veta}), where $Z= U^{2}$ as follows:
\begin{equation}
\begin{split}
f(z)&=\frac{z^{-\frac{1}{2}}(z^{\frac{1}{2}}+vt)}{r^{2}}\left(\frac{2}{\pi}\arccos\left(\frac{(z^{\frac{1}{2}}+vt)}{2r}\right)\right.- \\
&\left. \frac{(z^{\frac{1}{2}}+vt)}{\pi r}\sqrt{1-\frac{(z^{\frac{1}{2}}+vt)^2}{4 r^2}}\right).
\end{split}
\end{equation}
Let $Y=y(Z)$ be the derived RV and $h(y)$ be the PDF of the derived RV. We now apply the transformation of the density function to determine $h(y)$, which is defined as:
 \begin{align}\label{drift8}
\begin{split}
h(y)=f(z(y)) \left|\frac{dz}{dy}\right|
\end{split}
\end{align}
where $\frac{dz}{dy}=\frac{-4 Dt}{y}$. In order to find the PDF of $h(y)$ we find the inverse of $y(z)$, which results in $z(y)=(-4 D T log(\sqrt{4 \pi D t} \frac{y}{M}))$. Accordingly, we incorporate $z(y)$ in (\ref{drift8}) to find the expression in (\ref{drifft}).
\end{proof}
\section{Simulation Results}
\label{sim}
In this section, we conduct simulations to validate the analytical models found in the previous section. Our simulation setup consists of two points generated randomly within a circle of radius $r$. The Euclidean distance between these two points is then found and substituted in the received signal strength equations (\ref{diffu}) and (\ref{diffdrift}) for diffusion based communication and diffusion with drift respectively. This procedure is repeated $100,000$ times for values of $x$ in the range $[0,r]$.  The probability of each outcome is then calculated. We use Monte Carlo simulations to verify the robustness of the results.  
\subsection{Free Diffusion}

Fig.\ref{pdfr} (a) compares the analytical and simulation results for the case of the radius $r$ being equal to $3 mm$. Fig.\ref{pdfr} (b) conducts a similar comparison for the case of $r$ being equal to $8 mm$. The diffusion coefficient $D$ is considered of value $10^{-5} \frac{cm^{2}}{s}$ throughout the rest of the paper. It is observed that very good agreement between the simulation and analytical results is reported. In addition, it can be observed that as the radius increases smaller received signal strength values are obtained, which is expected, as there is a higher probability of loss when the distance between the nanonodes is larger.    
\begin{figure}\centering
\subfigure[r=3 mm]{\includegraphics[width=4cm]{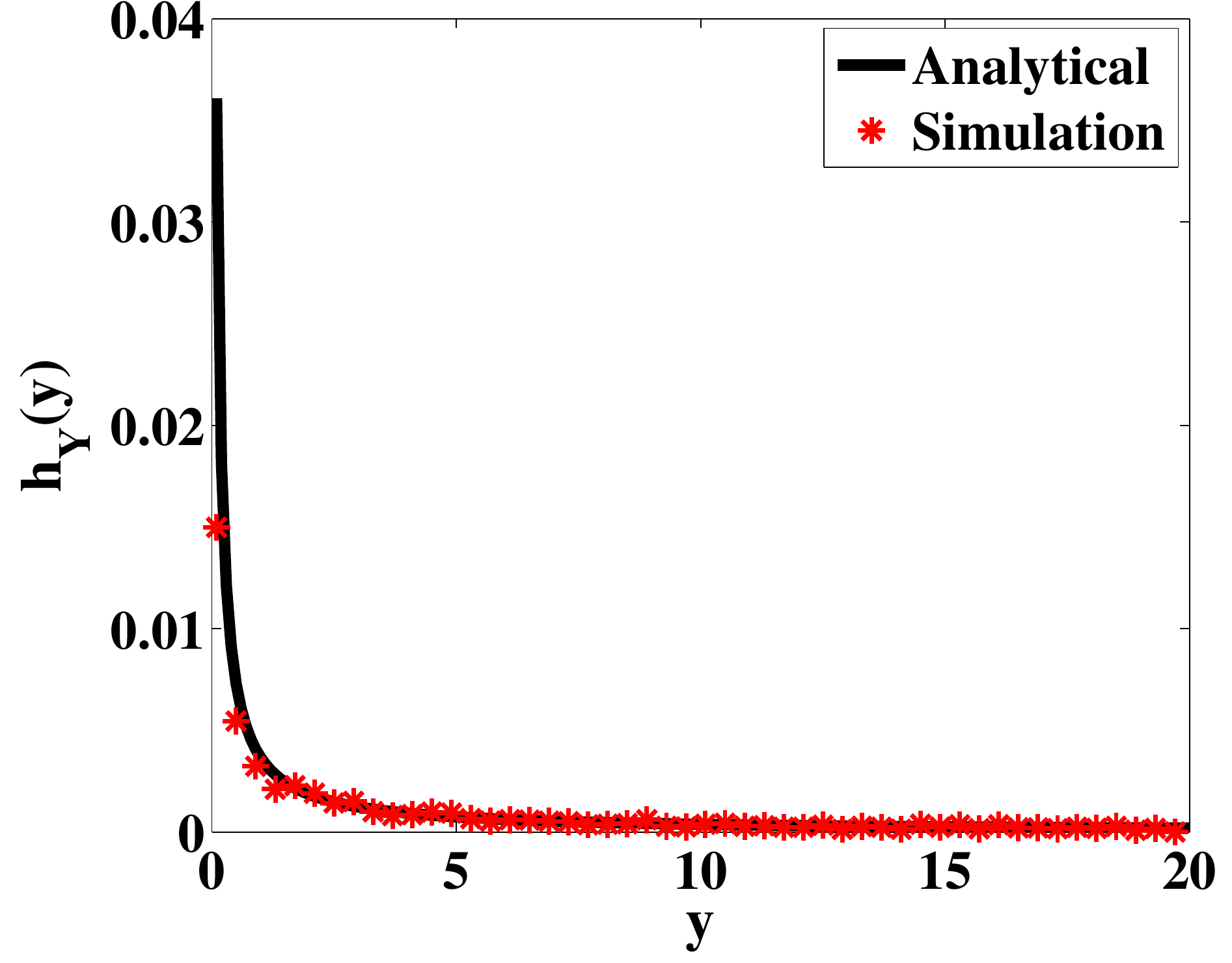}}
\centering
\hfill
\subfigure[r=8 mm]{\includegraphics[width=4cm]{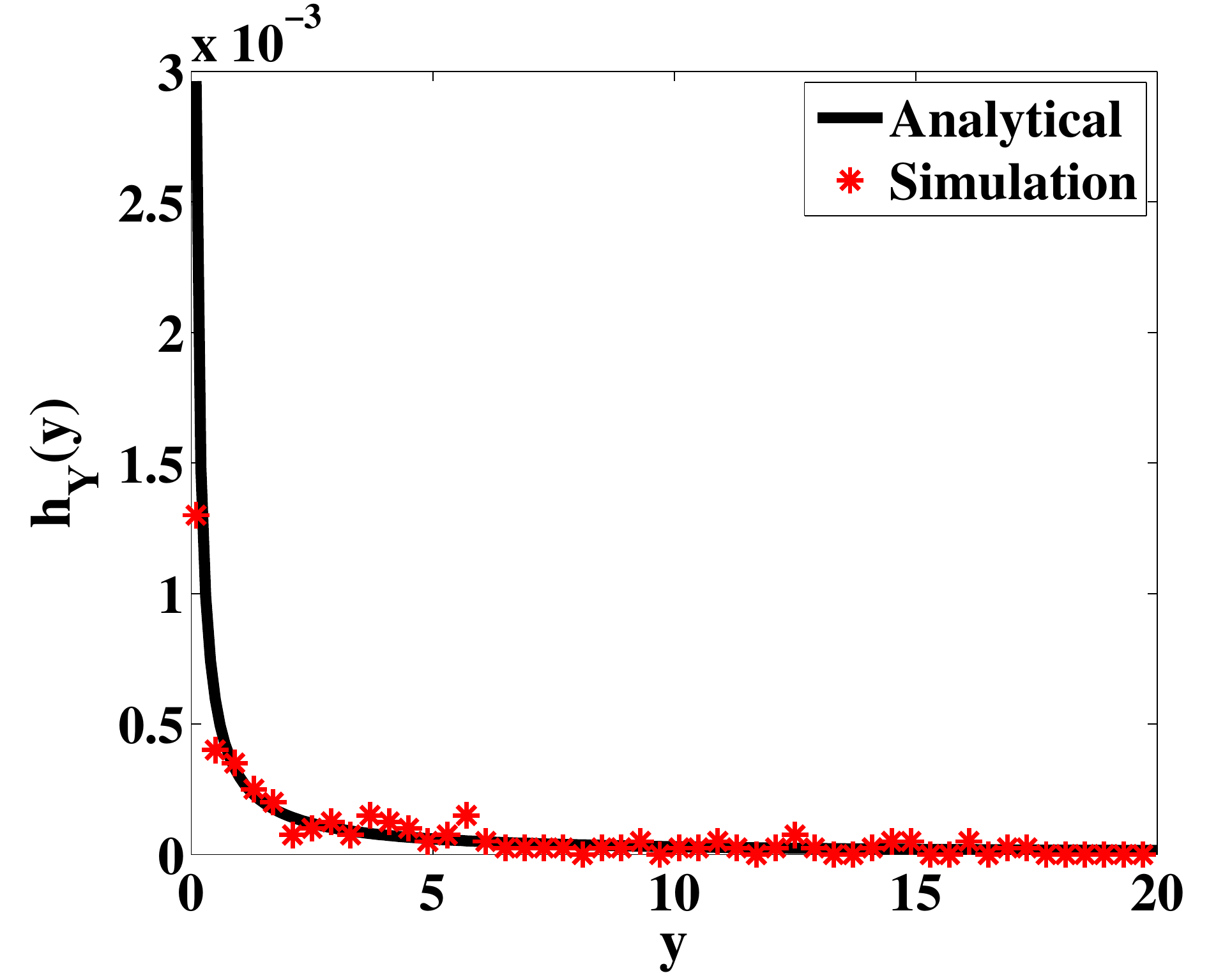}}
\hfill
\caption{PDF of the concentration for circular region of radius $r$}
\label{pdfr}
\end{figure}
As mentioned in the previous section, we use the PDF of the concentration to find the CDF, which defines the probability with which the concentration at the receiver will exceed a certain level when the number of molecules diffused at the transmitter is known. The CDF simulation results are found using the same method as the PDF results. Fig. \ref{cdfr} shows the CDF vs. values of $y$ between $1$ and $10$ for (a) $r=1.8$ $m$ and $M = 1$ and (b) $r=2.25$ $m$ and $M = 1$ at time $t=300s$. It can be observed that there is a good matching between the analytical and simulation results. The probability of success can be found from the CDF when a threshold for a successful transmission is defined. 
\begin{figure}\centering
\subfigure[r=1.8 m]{\includegraphics[width=5cm]{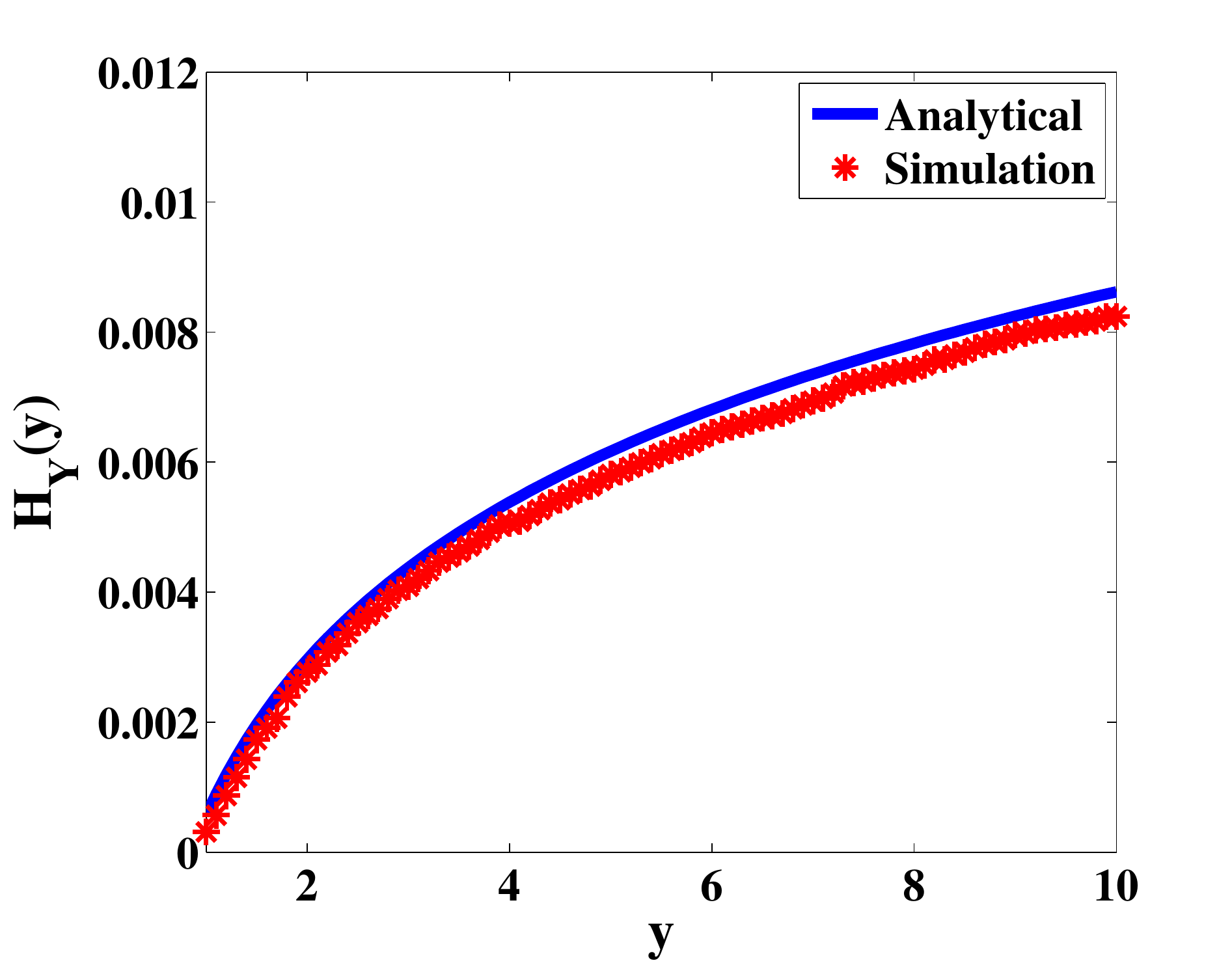}}
\centering
\hfill
\subfigure[r=2.25 m]{\includegraphics[width=5cm]{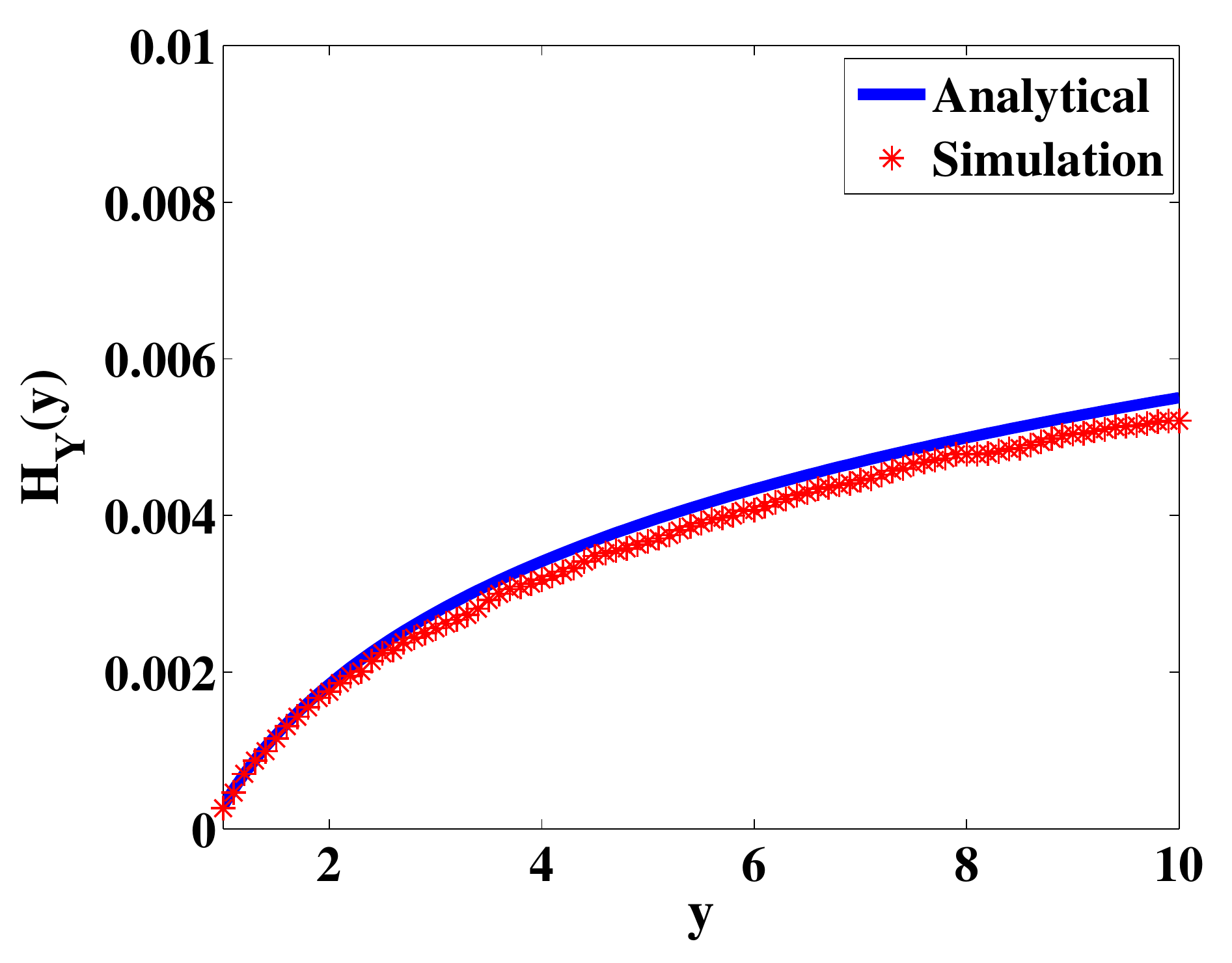}}
\hfill
\caption{CDF of the concentration for circular region of radius $r$}
\label{cdfr}
\end{figure}
Either the PDF or the CDF can be used to characterize the probabilities of successful message delivery at the receiver. Here, we indicatively show how such an analysis can be pursued for a simple receiver scheme. The scheme assumes successful reception of a $1$, if the concentration exceeds a particular threshold value $\tau$. The success probability is thus calculated by finding the probability of the received signal strength exceeding the threshold value. The CDF can be used to calculate the latter. Due to the approximations used in the formula of the CDF in (\ref{CDF11}) it has not been possible to validate it for every substitute. This numerical problem will be addressed in our future work. Nevertheless, we use numerical integration of the PDF to evaluate the success probability. We choose the threshold concentration value, after which a transmission is considered successful, $\tau$ to be $0.01$ $mol/mm^{2}$ and we plot the success probability versus the number of transmitted molecules for a particular time instance. This can be used as a guideline with which to decide how many molecules to transmit in order to achieve a particular success probability. We fix $r$ to $1.2$ mm and $t$ to $100$ seconds. The results in Fig. \ref{sucvsM} depict a concave type behavior; sharp increase at the beginning and saturation taking place at higher $M$ values, indicating a threshold value beyond which high success probability is achieved. We indicate on the diagram the value of $M$ that leads to a success probability of $0.9$. We then investigate how the value of $M$ that would suffice to achieve a particular success probability changes with the radius $r$. Again we consider a success probability of $90\%$. In Fig. \ref{radvsM} we plot $r$ as a function of $M$ with values of the latter ranging from $100$ molecules to $700$ whereas $r$ takes values from $1.1 mm - 1.5 mm$. It is worth noting that the values should have been in the micrometer scale, however, numerical problems have been reported when distances in micrometers are substituted. Thus, we use values in the millimeter region.
\label{stat}
\begin{figure}[h]
	\centering
		\includegraphics[width=5 cm]{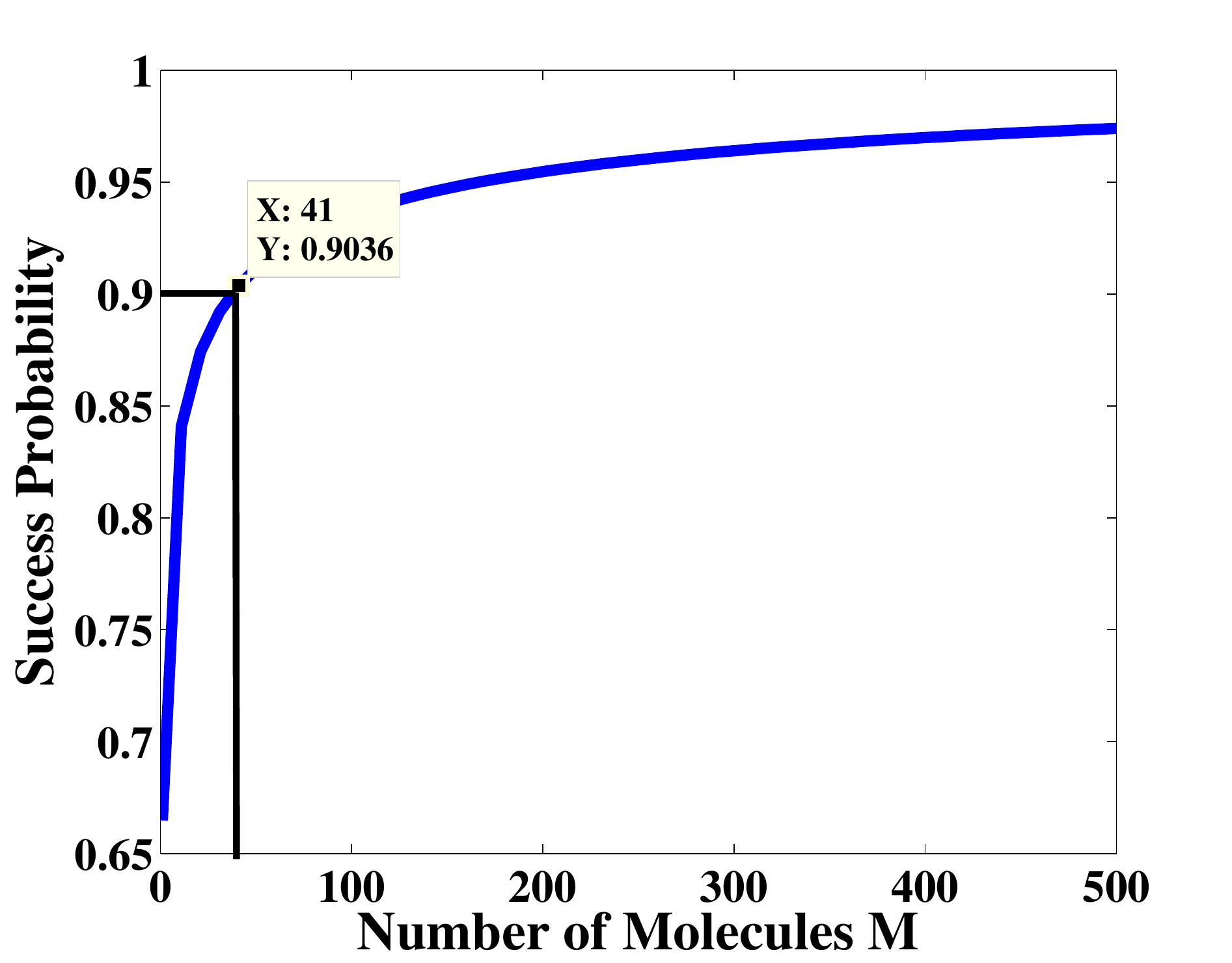}
	\caption{The probability of success versus the transmitted number of molecules }
	\label{sucvsM}
\end{figure}
\label{stat}
\begin{figure}[h]
	\centering
		\includegraphics[width=5 cm]{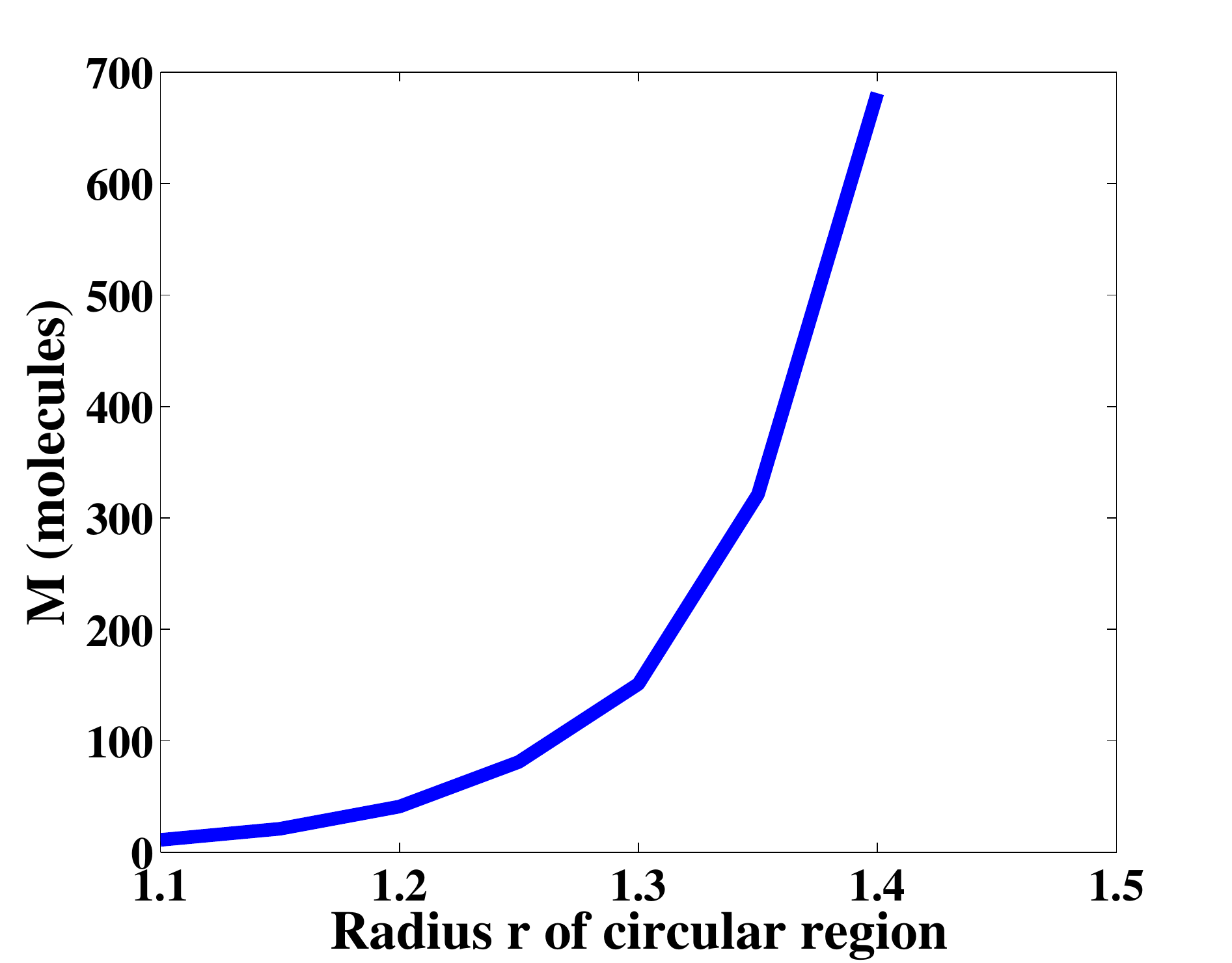}
	\caption{The threshold number of transmitted molecules versus the radius}
	\label{radvsM}
\end{figure}
\subsection{Diffusion with Drift}
We use similar simulation methodology to compare the PDF of the received signal strength in the case of diffusion with drift, when analysis and simulations are used. The difference is that we use equation (\ref{diffdrift}) instead of equation (\ref{diffu}). Fig. \ref{driftpdf} (a) compares analytical and simulation results in the case of $r$ being equal to $2 m$ whereas Fig. \ref{driftpdf} (b) compares the results in the case of $r=4 m$. Fairly good matching is observed between the two. The oscillatory behavior of the simulation results can be attributed to the singularities in the derived closed form expression, which render the computations challenging.  
\begin{figure}\centering
\subfigure[r=2 m]{\includegraphics[width=5cm]{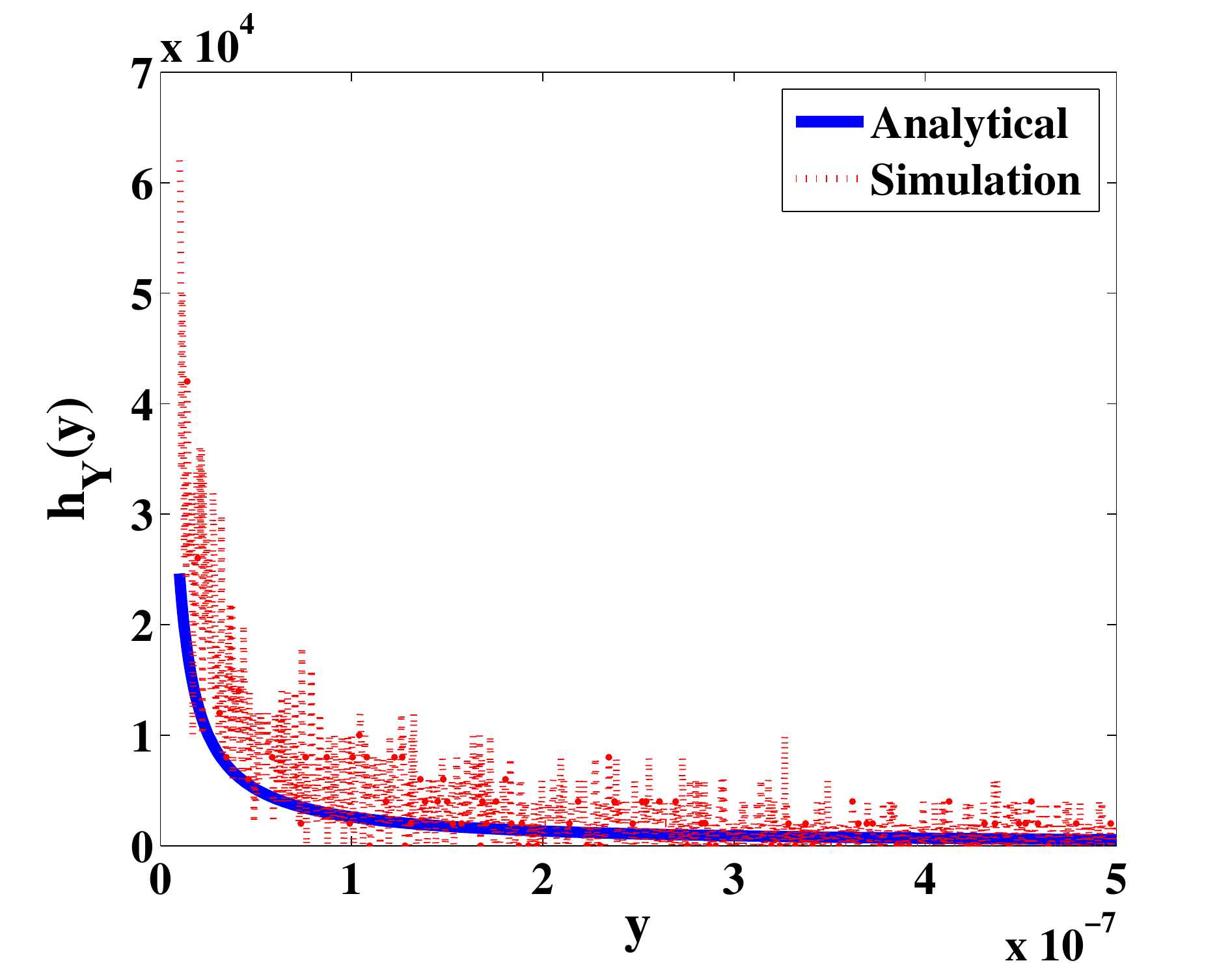}}
\centering
\hfill
\subfigure[r=4 m]{\includegraphics[width=5cm]{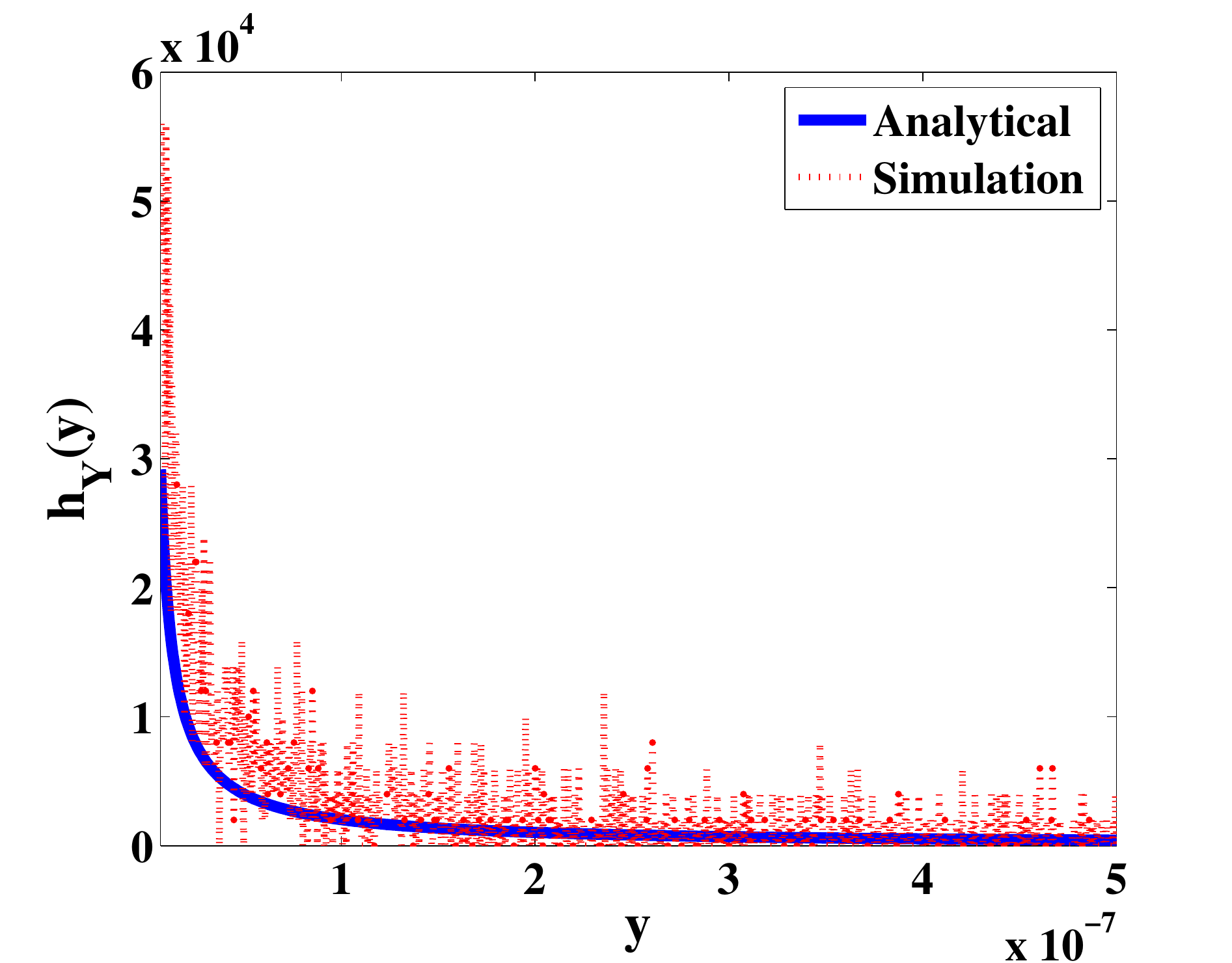}}
\hfill
\caption{PDF of diffusion with drift for circular region of radius $r$}
\label{driftpdf}
\end{figure}
%
%
\section{Conclusion} \label{conc}
In this paper we consider the molecular communication between a pair of nodes randomly distributed in a circular region. We review molecular communication methods and we derive the PDF of the received signal strength when free diffusion and diffusion with drift are employed. For the case of free diffusion the CDF is also derived and used to find desirable signal strengths when particular outage probabilities are required. The analytical results are compared with simulation results and good matching is observed between the two. In the future, we aim at calculating the PDF and CDF of the received signal strength for different communication methods, when multiple transmissions are conducted towards a particular receiver. This will allow for characterization of the SNR properties of the channel. It is also our aim to consider geometries beyond circular, as for example rectangular.    

\begin{acks}
This work was partially funded by the European Union via the Horizon 2020: Future Emerging Topics call (FETOPEN), grant EU736876, project VISORSURF (http://www.visorsurf.eu).
\end{acks}